\documentclass[draft,11pt]{article}
\usepackage{amssymb}
\usepackage{amsthm}
\usepackage{amsmath}
\usepackage{fullpage,hyperref}
\usepackage{graphicx,url,amsmath,amsthm,amsfonts,amssymb,subfigure,bbm}

\newcommand\remove[1]{}
\newcommand{\inter}{\mathrm{inter}}

\newcommand{\rnote}[1]{}

\newcommand{\supp}{\mathrm{supp}}
\newcommand{\dist}{\mathsf{dist}}

\newcommand{\R}{\mathbb{R}}

\newcommand{\len}{\mathsf{len}}

\newcommand{\1}{\mathbf{1}}

\newcommand{\cP}{\mathcal{P}}

\DeclareMathOperator{\diam}{diam}

\newtheorem{theorem}{Theorem}[section]
\newtheorem{lemma}[theorem]{Lemma}

\newtheorem{corollary}[theorem]{Corollary}

\newtheorem{fact}[theorem]{Fact}

\theoremstyle{definition}

\newtheorem{remark}{Remark}[section]

\newcommand{\vcon}{\mathsf{con}}
\newcommand{\val}{\mathsf{val}}
\DeclareMathOperator{\E}{\mathbb{E}}
\newcommand{\trans}{^{\!\top}}

\date{}

\begin{document}

\title{{\bf Eigenvalue bounds, spectral partitioning, \\ and metrical deformations via flows}}

\author{Punyashloka Biswal\footnote{Research supported by NSF CCF-0644037.
Department of Computer Science and Engineering, University of Washington,
    Seattle, WA.}
\and James R. Lee\footnotemark[\value{footnote}]
\and Satish Rao\thanks{Department of Computer Science, University of California, Berkeley, CA.}}

\maketitle

\begin{abstract}
We present a new method
for upper bounding the second eigenvalue of the
Laplacian of graphs.
Our approach uses multi-commodity flows to deform the geometry of the graph;
we embed the resulting metric into Euclidean space
to recover a bound on the Rayleigh quotient.
Using this, we show that every $n$-vertex
graph
of genus $g$ and maximum degree $d$ satisfies
$\lambda_2(G) = O(\frac{(g+1)^3 d}{n})$.
This recovers the $O(\frac{d}{n})$ bound of Spielman and Teng
for planar graphs, and compares to Kelner's bound of $O(\frac{(g+1) \mathrm{poly}(d)}{n})$,
but our proof does not make use of conformal mappings
or circle packings.  We are thus able to extend this to resolve
positively a conjecture of Spielman and Teng, by proving
that $\lambda_2(G) = O(\frac{d h^6 \log h}{n})$ whenever
$G$ is $K_h$-minor free.  This shows, in particular,
that spectral partitioning can be used to recover
$O(\sqrt{n})$-sized separators
in bounded degree graphs that exclude a fixed minor.
We extend this further by obtaining nearly optimal bounds on $\lambda_2$
for graphs which exclude small-depth minors in the sense of
Plotkin, Rao, and Smith.
Consequently, we show that spectral algorithms find small separators
in a general class of geometric graphs.

\smallskip

Moreover, while the standard ``sweep'' algorithm applied to the second eigenvector
may fail to find good quotient cuts in graphs of unbounded degree, our approach
produces a vector that works for {\em arbitrary} graphs.  This yields an alternate
proof of the result of Alon, Seymour, and Thomas
that every excluded-minor family of graphs has $O(\sqrt{n})$-node
balanced separators.
\end{abstract}


\section{Introduction}

Spectral methods are some of the most successful heuristics for graph
partitioning and its variants.  They have seen a great deal of success
in application domains such as mapping finite element calculations
onto parallel machines \cite{Simon91,Williams90}, solving sparse linear systems \cite{PSW92},
partitioning for domain decomposition \cite{CR87,CS93},  VLSI circuit
design and simulation \cite{CSZ93,HK92,AK95}, and image segmentation \cite{SM00}.
We refer to \cite{spielman-teng} for a discussion of their history and experimental
success.

Recent papers \cite{spielman-teng,GM98,Kelner06} have begun a theoretical analysis
of spectral partitioning for families of graphs on which it seems to work well
in practice.  Such analyses proceed by showing that the second eigenvalue
of the Laplacian of the associated graph is small; from this, one derives
a guarantee on the performance of simple spectral algorithms.
The previous approaches of Spielman and Teng \cite{spielman-teng} and Kelner \cite{Kelner06}
either work for graphs which already possess a natural geometric representation
(e.g. simplicial graphs or $k$-nearest-neighbor graphs), or use conformal mappings
(or their discrete analog, circle packings) to impart a natural geometric representation
to the graph.

Unfortunately, the use of these powerful tools makes it difficult
to extend their analysis to more general families of graphs.  We present
a new method for upper bounding the second eigenvalue of the
Laplacian of graphs.  As evidence of its efficacy, we resolve a conjecture
of Spielman and Teng about the second eigenvalue for excluded-minor families
of graphs.  Furthermore, we show that the ``spectral approach'' can be useful
for understanding the cut structure of graphs, even when spectral partitioning
itself may fail to find those cuts; this occurs mainly in the setting of graphs
with arbitrary degrees, and yields a new proof of the separator
theorem of Alon, Seymour, and Thomas \cite{AST90}.

\subsection{Previous results and our work}

Let $G=(V,E)$ be an $n$-vertex graph with maximum degree $d$.
Spielman and Teng \cite{spielman-teng} show that if $G$ is a planar graph,
then $\lambda_2 = O(d/n)$, where $\lambda_2$ is the second eigenvalue
of the Laplacian of $G$ (see Section \ref{sec:specprelim} for background
on eigenvalues and spectral partitioning).  It follows that a very simple
spectral ``sweep'' algorithm finds a quotient cut of ratio $O(d/\sqrt{n})$
in such graphs.  For $d=O(1)$, this shows that spectral methods can
recover the cuts guaranteed by the planar separator theorem of Lipton and Tarjan
\cite{LT79}; in particular, recursive bisection yields a balanced separator
which cuts only $O(\sqrt{n})$ edges.
The proof of Spielman and Teng is based on the Koebe-Andreev-Thurston circle packing theorem
for planar graphs, which provides an initial geometric representation of the graph.
Indeed, in his survey \cite{LovGRep07}, Lov\'asz notes that there is no known
method for proving the eigenvalue bound without circle packings.

In \cite{Kelner06}, Kelner proves that if $G$ is a graph of genus $g$,
then $\lambda_2 = O(\frac{g+1}{n}) \mathrm{poly}(d)$.  Again for $d=O(1)$,
this shows that spectral algorithms yield balanced separators of size $O(\sqrt{(g+1)n})$,
matching the bound of Gilbert, Hutchinson, and Tarjan \cite{GHT84}.  Kelner's proof
is not based on circle packings for genus $g$ graphs, but instead on the
uniformization theorem---the fact that every genus $g$ surface admits a certain
kind of conformal mapping onto the unit sphere.  (It turns out that the discrete
theory is not as strong in the case of genus $g$ circle packings.)
Kelner
must embed his graph on a surface, and then recursively subdivide the graph
(keeping careful track of $\lambda_2$), until it approximates the surface well
enough.

\medskip
\noindent
{\bf Excluded-minor graphs.}
The preceding techniques are highly specialized to graphs that can be endowed with some
conformal structure, and thus Spielman and Teng asked \cite{spielman-teng}
whether there is a
more combinatorial approach to bounding $\lambda_2$.  In particular, they
conjectured a significant generalization of the preceding results:
If $G$ excludes $K_h$ (the complete graph on $h$ vertices) as a minor,
then one should have $\lambda_2 = O(\frac{\mathrm{poly}(h)d}{n})$.  See Section
\ref{sec:minorprelim} for a brief discussion of graph minors.

Our new methods for bounding $\lambda_2$ are able to resolve this conjecture;
in particular, we prove that $\lambda_2 = O(\frac{h^6 (\log h) d}{n})$.
As a special case, this provides eigenvalue bounds in the planar and bounded
genus cases which bypass the need for circle packings or conformal mappings.
As stated previously, these bounds show that for $d,h=O(1)$, spectral
algorithms are able to recover the $O(\sqrt{n})$-sized balanced separators
of Alon, Seymour, and Thomas \cite{AST90} in $K_h$-minor-free graphs.

\medskip

\noindent
{\bf Geometric graphs.}
Spielman and Teng also bound $\lambda_2$ for geometric graphs,
e.g. well-shaped meshes and $k$-nearest-neighbor graphs in
a fixed number of dimensions.  Although these graphs do not exclude a $K_h$-minor
for any $h$ (indeed, even the $n \times n \times 2$ grid contains
arbitrarily large $K_h$ minors as $n \to \infty$), these graphs
do exclude minors {\em at small depth,} in the sense of Plotkin, Rao, and Smith \cite{PRS94}.
(Essentially, the connected components witnessing the minor must be of bounded diameter;
see Section \ref{sec:minorprelim}.)
Spielman and Teng \cite{spielman-teng}
ask whether one can prove spectral bounds for such graphs.

In Section \ref{sec:geographs},
we prove nearly-optimal bounds on $\lambda_2$ for graphs which exclude
small-depth minors.  This shows that
spectral algorithms can find small balanced separators for a large family of low-dimensional
geometric graphs.

\medskip
\noindent
{\bf Graphs with unbounded degrees.}
Finally, we consider separators in {\em arbitrary graphs,} i.e. without
imposing a bound on the maximum degree.  Very small separators can
still exist in such graphs, if we consider {\em node separators}
instead of the edge variety.  For example, Alon, Seymour, and Thomas \cite{AST90}
(following \cite{LT79,GHT84}) show that every $K_h$-minor-free graph
has a subset of nodes of size of $O(h^{3/2} \sqrt{n})$ whose removal breaks the graph
into pieces of size at most $n/3$.

The Laplacian of a graph is very sensitive
to the maximum degree, and thus one does not expect spectral partitioning
to do as well in this setting.
Nevertheless, we show that the ``spectral ideology''
can still be used to obtain separators in general.
We show that if one runs the ``sweep'' algorithm,
not on the second eigenvector of the Laplacian, but on the
vector we produce to bound the Rayleigh quotient,
then one recovers small separators regardless of the degree.
In particular, our approach is able to locate balanced node separators
of size $O(h^3\sqrt{\log h} \sqrt{n})$ in $K_h$-minor-free graphs;
this gives a new proof of the Alon-Seymour-Thomas result
(with a slightly worse dependence on $h$).

\medskip
\noindent
{\bf Overview of our approach.}
At a high level (discussed in more detail in Section \ref{sec:outline}), our approach
to bounding $\lambda_2$ proceeds as follows.  Given a graph $G$, we compute an
all-pairs multicommodity flow in $G$ which minimizes the $\ell_2$-norm
of the congestion at the vertices.  This flow at optimality is used
to deform the geometry of $G$ by weighting the vertices according
to their congestion.  We then embed the resulting vertex-weighted shortest path metric
into the line to recover a bound on the Rayleigh quotient, and hence on $\lambda_2$.
The remaining technical step is to get control on the structure of an optimal flow
in the various graph families that we care about.

We remark that our bounds are optimal, except for the slack that comes from the
embedding step.  E.g., for genus $g$ graphs we actually
achieve the bound $\lambda_2 = O(\frac{d g}{n}) \left(\min\{\log n, g\}\right)^2$,
where we expect that the latter factor can be removed.  For instance, our approach
might give a path toward improving the Alon-Seymour-Thomas separator result
to its optimal dependency on $h$.

\subsection{Preliminaries}
\label{sec:prelims}

Given two expressions $E$ and $E'$ (possibly depending on a number of parameters), we write $E = O(E')$ to mean that $E \leq C E'$
for some constant $C > 0$ which is independent of the parameters. Similarly, $E = \Omega(E')$ implies that $E \geq C E'$ for some $C > 0$.
We also write $E \lesssim E'$ as a synonym for $E = O(E')$.  Finally, we write $E \approx E'$ to denote
the conjunction of $E \lesssim E'$ and $E \gtrsim E'$.

All graphs in the paper are assumed to be undirected.  $K_n$ denotes the complete graph on $n$
vertices, and $K_{m,n}$ denotes the complete $m \times n$ bipartite graph.
For a graph $G$, we use $V(G)$ and $E(G)$ to denote the vertex and edge
sets of $G$, respectively.

\subsubsection{Eigenvalues and spectral partitioning}
\label{sec:specprelim}

Let $G=(V,E)$ be a connected graph with $n = |V|$.
The {\em adjacency matrix} $A_G$ of $G$ is an $n \times n$ matrix
with $(A_G)_{i,j} = 1$ if $(i,j) \in E$ and $(A_G)_{i,j} = 0$ otherwise.
The {\em degree matrix} of $G$ is defined by $(D_G)_{i,i} = \deg(i)$
for all $i \in V$, and $(D_G)_{i,j} = 0$ for $i \neq j$.
Finally, we define the {\em Laplacian of $G$} by
$$
L_G = D_G - A_G.
$$
It is easy to see that $L_G$ is a real, symmetric, positive semi-definite;
if we order the eigenvalues of $L_G$ as $\lambda_1 \leq \lambda_2 \leq \cdots \leq \lambda_n$,
and let $v_1, v_2, \ldots, v_n$ be a corresponding orthonormal basis of
eigenvectors, one checks that $\lambda_1 = 0$ and $v_1 = \frac{1}{\sqrt{n}} (1,1,\ldots,1)$.
A vast array of work in spectral graph theory relates the eigenvalues of $L_G$
to the combinatorial properties of $G$ (see, e.g. \cite{Chung97}).  In the present
work, we will be most interested in the connections between the second eigenvalue $\lambda_2$,
and the existence of small quotient cuts in $G$, following \cite{Cheeger70,AM85}.
We will write $\lambda_2(G)$ for $\lambda_2$ when $G$ is not clear from context.

Given a subset $S \subseteq V$, we define the {\em ratio} of the cut $(S,\bar S)$ by
\begin{equation}\label{eq:ratio}
\Phi_G(S) = \frac{|E(S,\bar S)|}{\min(|S|,|\bar S|)},
\end{equation}
where $E(S,\bar S)$ is the set of edges with exactly one endpoint in $S$.
We also define $\Phi^*(G) = \min_{S \subseteq V} \Phi_G(S)$.
Finally, we say that $S \subseteq V$ is a {\em $\delta$-separator}
if $\min(|S|,|\bar S|) \geq \delta n$.

Spectral partitioning uses the second eigenvector of $G$ to attempt to find
a cut with small ratio.  The most basic spectral partitioning algorithm uses the following
simple ``sweep.''
\begin{enumerate}
\item Compute the second eigenvector $z \in \mathbb R^n$ of $L_G$.
\item Order the vertices $V = \{1,2,\ldots,n\}$ so that $z_1 \leq z_2 \leq \ldots \leq z_n$
and output the cut of the form $\{1,  \ldots, i\}, \{{i+1},\ldots,n\}$
which has the smallest ratio.
\end{enumerate}

The next result is well-known and follows from the proof of the
Alon-Milman-Cheeger inequality for graphs \cite{Cheeger70,AM85};
see, e.g. \cite{spielman-teng,Mih89}.

\begin{theorem}
For any $v \in \mathbb R^n$ with $\sum_{i=1}^n v_i = 0$, the sweep algorithm
returns a cut $S \subseteq V$ with
$$
\Phi_G(S) \leq \sqrt{2d_{\max}\frac{\langle v, L_G v\rangle}{\|v\|^2}},
$$
where $d_{\max}$ is the maximum degree in $G$.
\end{theorem}

Furthermore, one can use recursive quotient cuts
to find small $\delta$-separators in $G$ \cite{LT79}.

\begin{lemma}
Let $G = (V,E)$, and suppose that
for any subgraph $H$ of $G$, we can find
a cut of ratio at most $\phi$.
Then a simple recursive quotient cut algorithm
returns a $\frac13$-separator $S \subseteq V$
with $|E(S,\bar S)| \leq O(\phi n)$.
\end{lemma}

\subsubsection{Graph minors}
\label{sec:minorprelim}

If $H$ and $G$ are two graphs, one says that
$H$ is a {\em minor} of $G$ if $H$ can be obtained from $G$
by a sequence of zero or more of the three operations:
edge deletion, vertex deletion, and edge contraction.
$G$ is said to be {\em $H$-minor-free} if it $H$
is not a minor of $G$.  We refer to \cite{Lo06,DiestelBook}
for a more extensive discussion of the vast graph minor theory.

Equivalently, $H$ is a minor of $G$
if there exists a collection of disjoint sets $\{ A_v \}_{v \in V(H)}$
with $A_v \subseteq V(G)$ for each $v \in V(H)$,
such that each $A_v$ is connected in $G$, and there is an
edge between $A_u$ and $A_v$ whenever $(u,v) \in E(H)$.

Following \cite{PRS94}, we see that {\em $H$ is a minor of $G$ at depth $L$}
if, additionally, there exists such a collection of sets with $\diam(A_v) \leq L$
for each $v \in V(H)$, where $\diam(A_v) = \max_{i,j \in A_v} \dist(i,j)$
and $\dist$ is the shortest-path distance in $G$.

\subsection{Outline}
\label{sec:outline}

We now explain an outline of the paper, as well as a sketch of our approach.
Let $G=(V,E)$ be a connected, undirected graph with $n=|V|$. Using the variational
characterization of the eigenvalues of $L_G$ (see \eqref{eq:variation}), we can write
$$\frac{\lambda_2(G)}{2n} =
\min_{f :V \to \mathbb R} \frac{\sum_{uv \in E} |f(u)-f(v)|^2}{\sum_{u,v \in V} |f(u)-f(v)|^2}
\geq \min_{d : V \times V \to \mathbb R_+}
\frac{\sum_{uv \in E} d(u,v)^2}{\sum_{u,v \in V} d(u,v)^2},
$$
where the latter minimum is over all {\em semi-metrics} on $V$,
i.e. all symmetric distance functions that satisfy the triangle inequality
and $d(u,u) = 0$ for $u \in V$.

Of course we are trying to prove {\em upper bounds} on $\lambda_2(G)$, but it is not
difficult to see that by Bourgain's theorem \cite{Bourgain85} on the embeddability
of finite metric spaces in Hilbert space, the second minimization is within an $O(\log n)^2$
factor of the first.  In Section \ref{sec:embeddings}, we discuss more refined notions
of ``average distortion'' embeddings which are able to avoid the $O(\log n)^2$ loss
for many families of graphs; in particular, we use the structure theorem of \cite{KPR93}
to achieve an $O(1)$ loss for excluded-minor families.

Thus we now focus on finding a semi-metric $d$
for which
\begin{equation}\label{eq:Rd}
\mathcal R_G(d) = \frac{\sum_{uv \in E} d(u,v)^2}{\sum_{u,v \in V}
d(u,v)^2}
\end{equation}
is small.  It is easy to see that for any graph $G$, the minimum
will be achieved by a shortest-path metric, and thus finding
such a $d$ corresponds to deforming the geometry of $G$ by
shrinking and expanding its edges. In actuality,
it is far more convenient to work with deformations
that involve {\em vertex weights,} but we use edge weights
here to keep the presentation simple.  Thus in the body of the paper,
all the edge notions expressed below are replaced by their vertex counterparts.

Unfortunately, $\min_d \mathcal R_G(d)$ is {\em not} a convex
optimization problem, so we replace it by the convexified
objective function
$$
\mathcal C_G(d) = \frac{\sqrt{\sum_{uv \in E} d(u,v)^2}}{\sum_{u,v \in V} d(u,v)}.
$$
In the proof of Theorem \ref{thm:eigenbound}, we connect $\mathcal R_G(d)$ and
$\mathcal C_G(d)$ via Cauchy-Schwarz; the structure of the extremal metrics
ensure that we do not lose too much in this step.

In Section \ref{sec:flows}, we show that minimizing $\mathcal C_G(d)$ {\em is} a convex
optimization problem, and thus we are able to pass to a dual formulation,
which is to send an all-pairs multicommodity flow in $G$, while
minimizing the $\ell_2$ norm of the congestion of the edges.
In fact, examination of the Lagrangian multipliers in
the proof of Theorem \ref{thm:duality} reveals that the optimal metric $d$
is obtained by weighting an edge proportional to its congestion in
an optimal flow.  Thus, by strong duality, in order to prove an upper bound on \eqref{eq:Rd}
for some graph $G$, it suffices to show that every all-pairs multicommodity
flow in $G$ incurs a lot of congestion in the $\ell_2$ sense.

We address this in Section \ref{sec:2-congest}.  First, we randomly round a fractional
flow to an integral flow, with only a mild blowup in the $\ell_2$-congestion.
In the case of planar (and bounded genus) graphs, we observe that an all-pairs
integral flow in $G$ induces a drawing of the complete graph in the plane.
By relating the $\ell_2$-congestion of the flow to the number of crossings
in this drawing, and using known results on graph drawings, we are able
to conclude that $\ell_2$-congestion must be large, finishing
our quest for upper bounds on the eigenvalues in such graphs (the
entire argument is brought together in Section \ref{sec:spectral}).

Extending this to $H$-minor-free graphs is more difficult, since there is no
natural notion of ``drawing'' to work with.  Instead, we introduce a generalized
``intersection number'' for flows with arbitrary demand graphs,
and use this in place of the crossing number in the planar case.
The intersection number is more delicate topologically, but after
establishing its properties, we are able to adapt the crossing number
proofs to establishing lower bounds on the intersection number,
and hence on the $\ell_2$-congestion of any all-pairs flow
in an excluded-minor graph.  We end Section \ref{sec:2-congest}
by extending our congestion lower bounds to graphs
which exclude small-depth minors.
This is important for the applications to geometric
graphs in Section \ref{sec:geographs}.

\medskip
\noindent
{\bf Balanced vertex separators with no $d_{\max}$ dependence.}
In the argument described above for bounding $\lambda_2(G)$, we lose
a factor of $d_{\max}$.  It turns out that if we simply want to
find a small {\em vertex separator} in $G$, then we can use the
vertex variant of the minimizer of \eqref{eq:Rd}
to obtain a metric on $G$, along with an appropriate
embedding of the metric into $\mathbb R$ from Section \ref{sec:embeddings}.
By passing these two components to the vertex-quotient cut rounding
algorithm of \cite{FHL05}, we are able to recover vertex separators
in arbitrary graphs, with no degree constraints.  This is carried
out in Section \ref{sec:separators}.

\subsection{Related work}

\medskip
\noindent
{\bf The Riemannian setting.}
Bounding the eigenvalues of the Laplace-Beltrami operator on Riemannian manifolds and, in particular,
surfaces, has a long and rich history in geometric analysis; see \cite{SY94}.
In particular, Hersch \cite{Hersch70} showed that for any Riemannian metric on the 2-sphere,
one has $\lambda_2(M) \leq O(\frac{1}{\mathrm{vol}(M)})$, where $\mathrm{vol}(M)$ denotes
the Riemannian volume of $M$.
The approach of Hersch has many parallels to that of Spielman and Teng,
and one can compare his bound to the $O(1/n)$ bound for $n$-node bounded-degree planar graphs.

Furthermore, Yang and Yau \cite{YY80} show that for a compact surface $M$ of genus $g$,
the bound $\lambda_2(M) \leq O(\frac{g+1}{\mathrm{vol}(M)})$ holds.
Of course, this is similar to Kelner's bound of $O(\frac{g+1}{n})$ for $n$-node graphs
of genus $g$; both proofs are based on conformal uniformization,
but Kelner's proof is more involved.  Indeed, the graph case is somewhat more
difficult since there are an infinite number of different topologies, while for
compact surfaces of genus $g$, there is only one.

Finally, Korevaar \cite{K93}, answering a question of Yau, gives
bounds on the higher eigenvalues of genus $g$ surfaces, of the form
$\lambda_k(M) \leq O(\frac{k(g+1)}{\mathrm{vol}(M)})$.
Grigor$'$yan and Yau \cite{GY99} discuss
some extensions of Korevaar's approach to bounding
eigenvalues of the Laplacian
on graphs, but their techniques require
the existence of a {\em very strong} volume measure on the graph,
e.g. in order to obtain our bounds for a $d$-regular graph,
they would require that for every vertex $x$ and $R \geq 1$,
$|B(x,R)| = O(R^2)$, where $B(x,R)$ is the $R$-ball about $x$.
We certainly cannot make such an assumption; in fact, the difficult
case for us is when the initial graph has very small diameter.

\medskip
\noindent
{\bf Connections with discrete conformal mappings.}
One can view the minimizer of \eqref{eq:Rd} (or,
more appropriately, the maximizer of the vertex version \eqref{eq:vertex})
as a sort of global ``uniformizing'' metric for general graphs.
In the setting of discrete conformal mappings,
a number of variationally defined objects appear, and
duality is often an important component in their analysis.
We mention, for instance, the {\em extremal length} \cite{Duffin62} as
a prominent example.  It also often happens that one chooses
a weight function $w : V \to \mathbb R_+$ as the minimizer
of some convex functional, and this weight function
plays the role of a discrete Riemannian metric (much
as is the case in Section \ref{sec:flows}); see, e.g. the
work of Schramm \cite{Schramm93} and He and Schramm \cite{HS95}.

A significant difference between these works and ours is
that the flow which is dual to the weight function
involves a single {\em commodity}, i.e. generally one node
is trying to send flow to one other node.  Our work is based
on the more global use of {\em multi-commodity flows,}
where the duality relationship is more complex (and, in particular,
a corresponding max-flow/min-cut theorem no longer holds).

\section{Metrics, flows, and congestion}
\label{sec:flows}

Let $G=(V,E)$ be an undirected graph, and for every pair $u,v \in
V$, let $\mathcal P_{uv}$ be the set of all paths between $u$ and $v$ in $G$.
Let $\mathcal P = \bigcup_{u,v \in V} \mathcal P_{uv}$.
A {\em flow in $G$} is a mapping $F : \mathcal P \to \mathbb R_+$.
We define, for every vertex $v \in V$, the value $$C_F(v) = \sum_{p \in \mathcal P : v \in p} F(p)$$ as
the {\em vertex congestion of $F$ at $v$.}
For $p \geq 1$, we define the
{\em vertex $p$-congestion of $F$} by $$\vcon_p(F) = \left( \sum_{v \in V} C_F(v)^p \right)^{1/p}.$$

We say that $F$ is an {\em integral flow} if, for every $u,v \in V$,
$|\{ p \in \mathcal P_{uv} : F(p) > 0\}| \leq 1$.
Given a {\em demand graph} $H=(U,D)$, we say that
$F$ is a {\em unit $H$-flow} if
there exists an injective mapping $g : U \to V$
such that for all $(i,j) \in D$, we have $\sum_{p \in \mathcal P_{g(i) g(j)}} F(p) = 1$,
and furthermore $F(p) = 0$ if $p \notin \bigcup_{(i,j) \in D} \mathcal P_{g(i) g(j)}$.
An {\em integral $H$-flow} is a unit $H$-flow which is also integral.
\remove{
Finally, we define the {\em value of $F$} as the quantity
$$
\val(F) = \min_{u,v \in V} \left(\sum_{p \in \mathcal P_{uv}} F(p)\right).
$$
...
}
\begin{lemma}\label{lem:RR}
For any graph $G=(V,E)$ and demand graph $H=(U,D)$, and
any unit $H$-flow $F$ in $G$, there exists an
integral $H$-flow $F^*$ such that
$$\vcon_2(F^*) \leq  \vcon_2(F) + \sqrt{\vcon_1(F)} \leq \vcon_2(F) + |V|^{3/2}.
$$
\end{lemma}

\begin{proof}
  For a flow $F \colon \cP \to \R_+$ and vertices $x,u,v$, let
  $F_{uv}(x) = \sum_{x \in p \in \cP_{uv}} F(p)$. Define the random
  flow $F^*$ as follows: For each demand pair $uv$, independently pick
  one path $p \in \cP_{uv}$ with probability $F(p)$. Set $F^*(p) = 1$
  for each of the selected paths, and zero for all other paths. Then
  \begin{align*}
    \E[\vcon_2(F^*)^2]
    &= \E\Bigl[ \sum_{x \in V} \Bigl( \sum_{u,v \in V} F_{uv}^*(x) \Bigr)^2 \Bigr]\\
    &= \sum_{x \in V} \Bigl( \sum_{u,v \in V} \E[F_{uv}^*(x)^2] +
    2 \sum_{\{u,v\}\neq\{u',v'\} \subseteq V} \E[F_{uv}^*(x)] \E[F_{u'v'}^*(x)] \Bigr).
    \intertext{Observing that $F^*_{uv}(x) \in \{0, 1\}$,}
    \E[\vcon_2(F^*)^2]
    &\leq \sum_{x\in V} \sum_{u,v \in V} \E[F_{uv}^*(x)] +
    \sum_{x \in V} \Bigl(\sum_{u,v \in V} \E[F_{uv}^*(x)]\Bigr)^2 \leq
    \vcon_1(F) + \vcon_2(F)^2.
  \end{align*}
  By concavity, we conclude that $\E[\vcon(F^*)] \leq
  \sqrt{\vcon_1(F)} + \vcon_2(F)$; in particular, there exists
  some fixed flow $F^*$ that achieves this bound.
\end{proof}

A non-negative vertex weighting $s : V \to \mathbb R_+$
induces a semi-metric $d_s : V \times V \to \mathbb R_+$ defined
by $d_s(u,v) = \min_{p \in \mathcal P_{uv}} \sum_{x \in p} s(x)$.
We define
\begin{eqnarray}\label{eq:vertex}
\Lambda_s(G) &=& \frac{\sum_{u,v\in V} d_s(u,v)}{\sqrt{\sum_{v \in V} s(v)^2}}.
\end{eqnarray}

The main theorem of this section follows.

\begin{theorem}[Duality of metrics and flows]
\label{thm:duality}
Let $G=(V,E)$ be any graph with $n=|V|$, then
$$
\min_F \vcon_2(F) = \max_{s : V \to \mathbb R_+} \Lambda_s(G),
$$
where the minimum is over all unit $K_n$-flows in $G$,
and the maximum is over all non-negative weight functions on $V$.
\end{theorem}

\begin{proof}
Let $P \in \{0,1\}^{\cP \times V}$ be the path incidence matrix
and $Q \in \{0,1\}^{\cP \times \binom{V}{2}}$ be the path endpoint
matrix, respectively, which are defined by
\begin{align*}
  P_{p,v} &=
  \begin{cases}
    1 & v \in p \\
    0 & \text{otherwise}
  \end{cases} &
  Q_{p,uv} &=
  \begin{cases}
    1 & p \in \cP_{uv} \\
    0 & \text{otherwise}.
  \end{cases}
\end{align*}
Then we write $\max_{s \colon V \to \R_+} \Lambda_s(G)$ as a convex
program \eqref{eq:maxflow} in standard form, with variables $(d,s) \in
\Omega = \R_+^{\binom{V}{2}} \times \R_+^V$.
\begin{equation}
  \label{eq:maxflow}\tag{\textsf{P}}
  \begin{array}{rll}
    \text{minimize} & -\1\trans d \\
    \text{subject to}
    & Qd \preceq Ps & \lVert s \rVert_2^2 \leq 1 \\
    & s \succeq 0 & d \succeq 0
  \end{array}
\end{equation}
Next, we introduce the Lagrangian multipliers $f \in \R_+^\cP$ and
$\mu \in \R_+$ and write the Lagrangian function
\begin{align*}
  L(d, s, f, \mu)
  &= -\1\trans d + f\trans(Qd - Ps) + \mu(s\trans s - 1) \\
  &= d\trans (Q\trans f - \1) + (\mu s\trans s - f\trans P s) - \mu.
\end{align*}
Therefore, the Lagrange dual $g(f, \mu) = \inf_{(d, s) \in \Omega}
L(d, s, f, \mu)$ is given by
\begin{align*}
  g(f,\mu)
  &= \inf_{d \succeq 0} d\trans (Q\trans f - \1) + \inf_{s \succeq 0} (\mu s\trans s -
  f\trans P s) - \mu.
\end{align*}
The dual program is then $\sup_{f, \mu} g(f,\mu)$. In order to write
it in a more tractable form, first observe that $g(f,\mu) = - \infty$
when $Q\trans f \prec \1$. But if we require that $Q\trans f \succeq
\1$, it is easy to see that the optimum must be attained when equality
holds. To minimize the quadratic part, set $\nabla (\mu s\trans s -
f\trans P s) = 0$ to get $s = P\trans f/2\mu$. With these
substitutions, the dual objective simplifies to
\begin{align*}
  g(f, \mu) &= - \frac{\lVert P\trans f \rVert_2^2}{4\mu} - \mu
\end{align*}
To maximize this quantity, set $\mu^* = \lVert P\trans f \rVert_2/2$, and
get $g(f, \mu) = - \lVert P\trans f \rVert_2$. Therefore, the final
dual program is
\begin{equation}
  \label{eq:dual} \tag{\textsf{P}*}
  \begin{array}{rll}
    \text{min} & \lVert P\trans f \rVert_2 \\
    \text{subject to} & f \succeq 0
    & Q\trans f = 1
  \end{array}
\end{equation}
When $P$ and $Q$ correspond to a $K_n$ demand graph for $G$, the dual
optimum is precisely $\min_f \vcon_2(f)$, where the minimum is over
unit $K_n$-flows.

The theorem now follows from Slater's condition in convex optimization;
see \cite[Ch.~5]{boyd}.

\begin{fact}[Slater's condition for strong duality]\label{lem:dual}
  When the feasible region for \eqref{eq:maxflow} has non-empty
  interior, the values of \eqref{eq:maxflow} and \eqref{eq:dual}
  are equal.
\end{fact}
\end{proof}

\section{2-congestion lower bounds}
\label{sec:2-congest}

In the present section, we prove lower bounds on the $2$-congestion
needed to route all-pairs multicommodity flows in various families of
graphs.  

\begin{theorem}[Bounded genus]
\label{thm:genus}
There exists a universal constant $c > 0$ such that
if $G=(V,E)$ is a genus $g$ graph with $n = |V|$, and $F$
is any unit $K_n$-flow in $G$, then $\vcon_2(F) \geq \frac{cn^2}{\sqrt{g}}$
for $n \geq 3 \sqrt{g}$.
\end{theorem}

\begin{proof}
By Lemma \ref{lem:RR} it suffices
to prove the theorem when $F$ is an integral flow.
Suppose, for the sake of contradiction,
there exists an integral $K_n$-flow $F$ with $\vcon_2(F) < \frac{n^2}{8\sqrt{g}}$.

The drawing of $G$ in a genus $g$ surface $\mathbb S$ induces (via $F$)
a drawing of $K_n$ in $\mathbb S$ where edges of $K_n$ only cross at (the images of)
vertices of $G$.
Clearly the number of crossings is upper bounded by $\sum_{v \in V} C_{F}(v)^2
= \vcon_2(F)^2 < \frac{n^4}{64g}$.  On the other hand,
it is known that as long as $n \geq 3 \sqrt{g}$, any drawing of $K_n$ in a
surface of genus $g$
requires at least $\frac{n^4}{64g}$ edge crossings \cite{acns, Leighton},
yielding a contradiction.
\end{proof}

Now we prove a similar theorem for $K_h$-minor-free graphs. To this
end, suppose we have a graph $G=(V,E)$ and an integral flow
$\varphi$ in $G$. For every $(i,j) \in E(H)$, let $\varphi_{ij}$ be
the corresponding flow path in $G$. Define
$$
\inter(\varphi) = \# \Bigl\{
(i,j),(i',j') \in E(H)
:
|\{i,j,i',j'\}|=4 \textrm{ and } \varphi_{ij} \cap \varphi_{i'j'} \neq \emptyset
\Bigr\}.
$$

\begin{lemma}\label{lem:contains}
If $\varphi$ is an integral $H$-flow in $G=(V,E)$ with $\inter(\varphi) = 0$ and
$H$ is bipartite with minimum degree 2, then
$G$ contains an $H$-minor.
\end{lemma}

\begin{proof}
Let $V(H) = L \cup R$ be a partition for which $E(H) \subseteq L \times R$.
Let $\varphi$ be an integral $H$-flow in $G$ with $\inter(\varphi)=0$.
We may assume that $V(H) \subseteq V$.
For a vertex $i \in V(H)$, let $N_H(i) \subseteq V_H$ be
the set of $j \in R$ with $(i,j) \in E(H)$.
For each $i \in V(H)$, let $V_i = \bigcup_{j \in N_H(i)} \varphi_{ij}$.

Now, for each $i \in L$ and $j \in N_H(i)$, consider
the path $\varphi_{ij} = \langle v_1, v_2, \ldots, v_k \rangle$,
and let $v_t$ be the first vertex in this path for which
$v_t \in \bigcup_{r \in L \setminus \{i\}} V_r$.
If no such $t$ exists, define $\hat \varphi_{ij} = \varphi_{ij}$,
and otherwise
define the prefix $\hat \varphi_{ij} = \{v_1, v_2, \ldots, v_{t-1}\}$.
Set $C_i = \bigcup_{j \in N_H(i)} \hat \varphi_{ij}$ to be
the union of all such prefixes.
Then, for each $i \in R$, define $C_i = V_i \setminus \bigcup_{j \in L} C_j$.

We claim that the sets $\{C_i\}_{i \in L \cup R}$ are all connected,
and pairwise disjoint, and that for $(i,j) \in E(H)$, we have
$E(C_i, C_j) \neq \emptyset$.  This will imply that $G$ has
an $H$-minor.
We start with the following straightforward fact.

\begin{fact}\label{fact:one}
If $v \in V_r \cap V_{r'}$ for some $r\neq r' \in L$, then $v \notin \bigcup_{j\in L} C_j$.
\end{fact}

\begin{lemma}\label{lem:self}
For every $i \in L \cup R$, we have $i \in C_i$.
\end{lemma}

\begin{proof}
First, we consider $i \in L$.
If $i \notin C_i$, then $i$ occurs
as an intermediate vertex of some $\varphi_{rs}$ path
for $r \in L, s \in R$ with $r \neq i$.
Since $i$ has degree at least 2 in $H$,
there must exist some $s' \in R$ with $s \neq s'$
and $(i,s') \in E(H)$.  But now $i \in \varphi_{rs} \cap \varphi_{is'}$
which contradicts the fact that $\inter(\varphi)=0$.
Thus we must have $i \in C_i$.

To see that $i \in C_i$ for $i \in R$,
note that by assumption $\deg_H(i) \geq 2$,
so there must exist $r \neq r' \in L$ for which
$(r,i), (r',i) \in E(H)$.  Thus
$i \in V_r \cap V_{r'}$, and by Fact \ref{fact:one},
it must be that $i \notin \bigcup_{j \in L} C_j$.
We conclude that $i \in C_i$.
\end{proof}

\noindent
{\bf Connected and disjoint components.}
Lemma \ref{lem:self} implies that $i \in C_i$ for $i \in L$,
so it is clear by construction that the sets $\{C_i\}_{i \in L}$
are each connected and that for any $i \in L$ and
$j \in L \cup R \setminus \{i\}$, we have $C_i \cap C_j = \emptyset$.
Thus we need only verify that each set $C_i$ is connected
for $i \in R$, and also that for $i,j \in R$ with $i\neq j$,
we have $C_i \cap C_j = \emptyset$.

\begin{lemma}\label{lem:complement}
If $i \in R$ and $j \in N_H(i)$,
then $\varphi_{ji} \setminus \hat \varphi_{ji} \subseteq C_i$.
\end{lemma}

\begin{proof}
Any node $v \in \varphi_{ji} \setminus \hat \varphi_{ji}$
must be contained either in $C_i$ or in $C_r$ for some $r \in L$
with $r \neq j$.  But the latter case cannot occur
because any node which is contained in $V_j \cap V_{r}$ for $r \neq j$
cannot be contained in $C_r$ by Fact \ref{fact:one}.
\end{proof}

Using the fact that $i \in C_i$ (Lemma \ref{lem:self})
and the preceding lemma, we see that $C_i$ is connected
for every $i \in R$.
It remains to show that for $i,j \in R$ with $i \neq j$,
we have $C_i \cap C_j = \emptyset$.

Suppose, to the contrary, that $C_i \cap C_j \neq \emptyset$.
Since $\inter(\varphi) = 0$, there must exist a $k \in L$
such that $(\varphi_{ki} \cap C_i) \cap (\varphi_{kj} \cap C_j) \neq \emptyset$.
The following lemma shows this to be impossible.

\begin{lemma}
For $k \in L$ and $i,j \in N_H(k)$, we must have
$\varphi_{ki} \cap \varphi_{kj} \subseteq C_k$.
\end{lemma}

\begin{proof}
Suppose, to the contrary, that there is a $v \in \varphi_{ki} \cap \varphi_{kj}$
for which $v \notin C_k$.  In this case, it must be that $v \in V_r$
for some $r \in L$ with $r \neq k$.  In other words, for some $s \in R$,
$\varphi_{rs}$
intersects both $\varphi_{ki}$ and $\varphi_{kj}$, but this is impossible
since $i \neq j$ and $\inter(\varphi) = 0$.
\end{proof}

\noindent
{\bf Edges of $E(H)$.}
Consider $i \in L$ and $j \in R$ with $(i,j) \in E(H)$.
It is straightforward to see that $\hat \varphi_{ij} \subseteq C_i$
by construction, and on the other hand, $\varphi_{ij} \setminus \hat \varphi_{ij} \subseteq C_j$,
by Lemma \ref{lem:complement}.  Since $i \in C_i$ and $j \in C_j$ by Lemma \ref{lem:self},
it follows that $E(C_i, C_j) \neq \emptyset$.  This complete the proof.
\end{proof}

\begin{corollary}
For every $h \geq 2$,
if $G$ is $K_h$-minor-free, and $\varphi$ is an integral $K_{2h}$-flow in $G$,
then $\inter(\varphi) > 0$.
\end{corollary}

\begin{proof}
If $\varphi$ is an integral $K_{2h}$ flow with $\inter(\varphi)=0$,
the obviously it induces an integral $K_{h,h}$ flow with the same property.
By Lemma \ref{lem:contains}, $G$ has $K_{h,h}$ as a minor, and hence also
has $K_h$ as a minor,
yielding a contradiction.
\end{proof}

\begin{lemma}\label{lem:inter}
For any integral flow $\varphi$ in $G$, we have $\vcon_2(\varphi) \geq \sqrt{\inter(\varphi)}$.
\end{lemma}

\begin{proof}
Clearly, $$\mathrm{inter}(\varphi) \leq \sum_{v \in V} \left(\sum_{i,j,i',j' \in E(H)} {\bf 1}_{v \in \varphi_{ij}}
\cdot {\bf 1}_{v \in \varphi_{i'j'}}\right)
= \sum_{v \in V} C_\varphi(v)^2. \qedhere$$
\end{proof}

We begin with an elementary proof that yields a suboptimal bound dependence on $h$.

\begin{theorem}\label{thm:minor}
If $G=(V,E)$ is $K_h$-minor-free and $n=|V|$, then
any unit $K_n$-flow $F$ in $G$ has $\vcon_2(F) \geq  \frac{n^2}{12 h^{3/2}}$
for $n \geq 4h$.
\end{theorem}

\begin{proof}
Using Lemma \ref{lem:RR}, it suffices to prove the theorem when $F$ is an integral $K_n$-flow in $G$.
By Lemma \ref{lem:inter}, it suffices
to show that $\mathrm{inter}(F) \geq \frac{n^4}{16 h^3}$.

If $\varphi$ is any integral flow with $\mathrm{inter}(\varphi) > 0$, then
one can always remove a terminal of $\varphi$
to obtain an integral flow $\varphi'$ for which $\mathrm{inter}(\varphi') \leq \mathrm{inter}(\varphi) - 1$.
From Lemma \ref{lem:contains}, we know
that for an integral $K_{2h}$-flow $\varphi$, we have $\inter(\varphi) > 0$.
It follows that if $\varphi$ is an integral $K_r$-flow in $G$,
then $\mathrm{inter}(\varphi) \geq r-2h+1$.

\medskip

Now let $p \in [0,1]$, and consider choosing a random subset $S_p \subseteq V$ by
including every vertex independently with probability $p$.  Let $n_p = |S_p|$,
and let $F_p$ be the integral $K_{n_p}$-flow formed by restricting the terminals
of $F$ to lie in $S_p$.  It is obvious that $\mathbb E[n_p] = pn$ and $\mathbb E[\mathrm{inter}(F_p)] = p^4 \cdot \mathrm{inter}(F)$,
since all intersections counted by $\mathrm{inter}(F)$ involve four distinct vertices.
Hence,
\begin{equation}\label{eq:Hminor}
p^4 \cdot \mathrm{inter}(F) = \mathbb E[\mathrm{inter}(F_p)] \geq \mathbb E[n_p - h +1] \geq  pn - 2h.
\end{equation}

We may assume that $n \geq 4h$, and in this case
choosing $p = \frac{4h}{n}$ in \eqref{eq:Hminor} yields
$$
\mathrm{inter}(F) \geq 2h \left(\frac{n}{4h}\right)^4 =\frac{n^4}{128 h^3}.
$$
finishing the proof.
\end{proof}

To do better, we first require the following theorem proved
independently by Kostochka \cite{Kost82} and Thomason \cite{Thom84}.
\begin{theorem}
\label{thm:avgdegree}
There exists a constant $c_K \geq 1$ such that for every $h \in \mathbb N$,
every graph $H$ with $|E(H)| \geq c_{KT} |V(H)| h \sqrt{\log h}$ contains
$K_h$ as a minor.
\end{theorem}

We remark that the preceding theorem is tight \cite{Kost82, Vega83, BCE80}.
We now proceed to an improved bound.

\begin{theorem}[Excluded minors]
\label{thm:exminors}
There exists a universal constant $c > 0$ such that
if $G=(V,E)$ is $K_h$-minor-free and $n=|V|$, then
any unit $K_n$-flow $F$ in $G$ has $\vcon_2(F) \geq  \frac{cn^2}{h \sqrt{\log h}}$
for $n \geq 4 c_{KT} h \sqrt{\log h} +1$, where $c_{KT}$
is the constant from Theorem \ref{thm:avgdegree}.
\end{theorem}

\begin{proof}
As in the proof Theorem \ref{thm:minor}, it suffices to prove that
$\mathrm{inter}(F) = \Omega(\frac{n^4}{h^2 \log h})$ whenever
$F$ is an integral $K_n$-flow in $G$.

If $\varphi$ is an integral $H$-flow with $\inter(\varphi) > 0$, then
obviously there exists an edge $e \in E(H)$ and an integral $(H \setminus e)$-flow
$\varphi'$ for which $\inter(\varphi') \leq \inter(\varphi) - 1$.
Combining this with Theorem \ref{thm:avgdegree} and Lemma \ref{lem:contains}
shows that for any $H$-flow $\varphi$ in $G$, we have
$$\inter(\varphi) \geq |E(H)| - 2c_{KT} |V(H)| h \sqrt{\log h},$$
where $c_{KT}$ is the constant from Theorem \ref{thm:avgdegree}.

We now apply this to the $K_n$-flow $F$.
As in the proof of Theorem \ref{thm:minor},
let $p \in [0,1]$, and consider choosing a random subset $S_p \subseteq V$ by
including every vertex independently with probability $p$.  Let $n_p = |S_p|$,
and let $F_p$ be the integral $K_{n_p}$-flow formed by restricting the terminals
of $F$ to lie in $S_p$.  We have,
\begin{equation}\label{eq:Hminor2}
p^4 \cdot \mathrm{inter}(F) = \mathbb E[\mathrm{inter}(F_p)] \geq p^2 \frac{n(n-1)}{2} - 2 c_{KT} pnh \sqrt{\log h}
\end{equation}

We may assume that $n \geq 4 c_{KT} h \sqrt{\log h} + 1$, and in this case
choosing $p = \frac{4 c_{KT} h \sqrt{\log h}}{n-1}$ in \eqref{eq:Hminor2} yields
$$
\mathrm{inter}(F) \geq \frac{n(n-1)^3}{64 c_{KT}^2 h^2 \log h},
$$
finishing the proof.
\end{proof}

\paragraph*{Bounds for shallow excluded minors.}
Finally, we prove congestion lower bounds for graphs which
exclude minors at small depth.  This is useful for applications
to geometric graphs in Section \ref{sec:geographs}.

\begin{theorem} \label{thm:shallowcon}
There exists a constant $c > 0$ such that if $G=(V,E)$
excludes a $K_h$-minor at depth $L$ and $n=|V|$, then
any unit $K_n$-flow in $G$ has $\vcon_2(F) \geq c \min\left(\frac{n^2}{h\sqrt{\log h}}, n^{3/2} L\right)$
for $n \geq c h \sqrt{\log h}$.
\end{theorem}

\begin{proof}
Suppose that $G$ excludes a $K_h$-minor at depth $L$, and let $F$ be an integral
$K_n$ flow in $G$.  First, we state the following straightforward
strengthening of Lemma \ref{lem:contains}.

\begin{lemma}\label{lem:containsL}
If $\varphi$ is an $H$-flow in $G$ in which every flow path has length at most $L/2$,
then $G$ contains a depth-$L$ $H$-minor.
\end{lemma}

Now, if at least half of the $\frac{n(n-1)}{2}$ flow paths in $F$ have length greater than $L/2$,
then the total length of flow paths is at least $\Omega(n^2) L$, which shows
that $$\vcon_2(F) = \sqrt{\sum_{v \in V} C_F(v)^2} \geq n^{-1/2} \sum_{v \in V} C_F(v) = \Omega(n^{3/2}) L.$$
If, on the other hand, at least half of the flow paths in $F$ have length less than $L/2$,
let $F'$ be the flow restricted to such paths.  Clearly $F'$ is an integral $H$-flow for some
dense graph $H$ on $n$ nodes, hence the proof of Theorem \ref{thm:exminors} (with
Lemma \ref{lem:containsL} substituted for Lemma \ref{lem:contains})
shows that for $n$ large enough, we have $\vcon_2(F) \geq \vcon_2(F') = \Omega\left(\frac{n^2}{h \sqrt{\log h}}\right)$.
\end{proof}

\begin{corollary}
There exists a constant $c > 0$ such that the following holds.
Suppose that for some $d \geq 1$ and every $L \geq 1$,
$G=(V,E)$ excludes a $K_{L^d}$ minor at depth $L$.
If $n=|V|$, then $\vcon_2(F) \geq c n^{3/2} \left(\frac{n}{d}\right)^{1/(2d+2)}.$
\end{corollary}


\section{Average distortion embeddings and random partitions}
\label{sec:embeddings}

In this section, we use a construction of Rabinovich to embed certain
metrics into the line with small ``average distortion.''
This allows us to pass from a good metric on a graph
to a good bound on the Rayleigh quotient.
Our main technique is the use of random padded partitions,
a now standard tool in the construction of metric embeddings (see, e.g. \cite{Bartal96,Rao99,Rab03,KLMN05}).

\subsection{Random partitions}

Let $(X,d)$ be a finite metric space.
We recall the standard definitions for padded decompositions (see, e.g. \cite{KLMN05}).
If $P$
is a partition of $X$, we will also consider
it as a function $P : X \to 2^X$ such that
for $x \in X$, $P(x)$ is the unique $C \in P$
for which $x \in C$.

Let $\mu$ be a distribution over
partitions of $X$,
and let $P$ be a random partition
distributed according to $\mu$.
We say that $P$ is $\Delta$-bounded if
it always holds that for $S \in P$,
$\diam(S) \leq \Delta$.

Given a $\Delta$-bounded random partition $P$,
we say that $P$ is $\alpha$-padded if
for every $x \in X$, we have
$$
\Pr\left[ B(x,\Delta/\alpha) \subseteq P(x)\right] \geq \tfrac12,
$$
where $B(x,r) = \{ y : d(x,y) \leq r \}$ denotes
the closed ball of radius $r$ about $x$.

We recall that the {\em modulus of padded decomposability}
is the value
$$\alpha(X,d) = \sup_{\Delta \geq 0} \left\{ \vphantom{\bigoplus}
\alpha : \textrm{$X$ admits a $\Delta$-bounded
$\alpha$-padded random partition} \right\}.$$

Now we can state a consequence \cite{Rao99} of the main theorem
of Klein, Plotkin, and Rao \cite{KPR93}.

\begin{theorem}\label{thm:KPR}
Let $G = (V,E)$ be any graph which excludes $K_r$ as a minor,
let $\len :E \to \mathbb R_+$ be an assignment
of non-negative lengths to the edges of $G$,
and let $d_G$ be the associated shortest-path semi-metric on $G$.
Then $\alpha(V,d_G) = O(r^2)$.
\end{theorem}

\begin{corollary}\label{cor:KPRvertex}
If $G = (V,E)$ excludes $K_r$ as a minor,
and $s : V \to \mathbb R_+$ is a non-negative weight
function on vertices, then the induced semi-metric $d_s$
satisfies $\alpha(V,d_s) = O(r^2)$.
\end{corollary}

\begin{proof}
Simply define $\len(u,v) = s(u)+s(v)$ for $(u,v) \in E$.  Clearly
the shortest-path distances induced by $\len$ and $s$ are within
a factor of 2, so the result follows from Theorem \ref{thm:KPR}.
\end{proof}

We also have the following theorem of Bartal for general metrics \cite{Bartal96}
(see \cite{KLMN05} for a proof of the precise statement, based on \cite{CKR01}).

\begin{theorem}\label{thm:bartal}
For any finite metric space $(X,d)$, we have $\alpha(X,d) = O(\log |X|)$.
\end{theorem}

\subsection{Average distortion embeddings}

\remove{
\begin{theorem}
\label{thm:bourgain}
For every $n$-point metric space $(X,d)$, there exists
a non-expansive mapping $f : X \to \mathbb R$ with
$$\sum_{u,v \in V} d(u,v)^p \lesssim O(\log n)^p \sum_{u,v \in V} (f(u)-f(v))^p.$$
\end{theorem}
}

Rabinovich \cite{Rab03} essentially proved the following theorem in the case $p=1$.
For applications to eigenvalues, the case $p=2$ is of particular interest (see Section \ref{sec:eigenvaluebounds}).
For direct application to vertex separators (Section \ref{sec:separators}), we will employ the $p=1$ case.
For a metric space $(X,d)$, a mapping $f : X \to \mathbb R$ is said to be {\em non-expansive}
if $|f(x)-f(y)| \leq d(x,y)$ for all $x,y \in X$.

\begin{theorem}
\label{thm:embedding}
For every $p \geq 1$, there exists a constant $C_p \geq 1$ such
that
for any metric space $(X,d)$, there exists
a non-expansive mapping $f : X \to \mathbb R$ with
$$\sum_{u,v \in X} d(u,v)^p \leq C_p \left[\alpha(X,d)\right]^p \sum_{u,v \in X} |f(u)-f(v)|^p.$$
\end{theorem}

\begin{proof}
We prove the theorem for $p=2$.  The other cases are similar.
Let $\Delta_2 = \sqrt{\frac{1}{n^2} \sum_{u,v \in X} d(u,v)^2}$.
First, we handle the case when many points are clustered about
a single node $x_0 \in X$.
In what follows, we use $B(x,R) = \{ y \in X : d(x,y) \leq R \}$
to denote the closed ball of radius $R$ about $x \in X$.

\bigskip
\noindent
{\bf Case I: There exists $x_0 \in X$ for which $|B(x_0, \frac14 \Delta_2)| \geq \frac{n}{10}$.}

\medskip
\noindent
In this case, let $S = B(x_0, \frac14 \Delta_2)$, and define $f(u) = d(u,S)$.
First, we have
\begin{eqnarray*}
n^2 \Delta_2^2 = \sum_{u,v \in X} d(u,v)^2 &\leq& 2 \sum_{u,v} \left[d(u,x_0)^2 + d(v,x_0)^2\right] \\
&=& 4n \sum_{u \in X} d(u,x_0)^2 \\
&\leq& 4n \sum_{u \in X} \left[ d(u,S) + \frac{\Delta_2}{4}\right]^2 \\
&\leq& \frac{n^2 \Delta_2^2}{2} + 8n \sum_{u \in X} d(u,S)^2.
\end{eqnarray*}
Therefore, $\sum_{u \in X} d(u,S)^2 \geq \frac{n \cdot \Delta_2^2}{16}.$
We conclude that
\begin{eqnarray*}
\sum_{u,v \in X} (f(u)-f(v))^2 &=& \sum_{u,v \in X} \left[ d(u,S) - d(v,S) \right]^2 \\
&\geq&
\sum_{u \notin S, v \in S} \left[d(u,S) - d(v,S)\right]^2 \\
&=&
\sum_{u \notin S, v \in S} d(u,S)^2 \\
&=&
|S| \sum_{u \in X} d(u,S)^2 \\
&\gtrsim &
n^2 \Delta_2^2 \\
&\gtrsim& \sum_{u,v\in X} d(u,v)^2.
\end{eqnarray*}
This finishes the clustered case.

\bigskip
\noindent
{\bf Case II: For every $u \in X$, $|B(u, \frac14 \Delta_2)| < \frac{n}{10}$.}

\medskip
\noindent
In particular, we know that for any subset $T \subseteq X$
with $\diam(T) \leq \frac14 \Delta_2$, we have $|T| < n/10$.

Now, let $P$ be a random partition of $X$ which is
$\frac14 \Delta_2$ bounded and $\alpha$-padded, where
$\alpha = \alpha(X,d)$.
We know that for every $x \in X$, we have
$$
\Pr\left[B(x, \Delta_2/(4\alpha)) \subseteq P(x)\right] \geq \frac12.
$$
So by Markov's inequality, it must be that there exists a partition
$P_0$ such that the set $$H_0 = \left\{ x \in X : B(x,\Delta_2/(4\alpha)) \subseteq P(x)\right\}$$
has $|H_0| \geq n/2$.  Fix this choice of $P_0$ and $H_0$.

Let $\{\sigma_C\}_{C \in P_0}$
be a collection of i.i.d. uniform 0/1 random variables, one for each
cluster $C \in P_0$ and define $S = \bigcup_{C \in P_0 : \sigma_C = 0} C$.
Finally, define $f : X \to \mathbb R$
by $f(u) = d(u,S)$.

Note that $f$ is a random function.
We will now argue that
\begin{equation}\label{eq:central}
\mathbb E \left[ \sum_{u,v \in X} (f(u)-f(v))^2 \right] \gtrsim \left(\frac{n\Delta_2}{\alpha}\right)^2
\gtrsim \alpha^{-2} \sum_{u,v\in X} d(u,v)^2,
\end{equation}
which will imply (by averaging)
that there exists a choice of $f : X \to \mathbb R$
for which the sum is at least $\Omega(\alpha^{-2}) \sum_{u,v\in X} d(u,v)^2.$

So it remains to prove \eqref{eq:central}.  Note that for every $C \in P_0$,
we have $\diam(C) \leq \Delta_2/4$, so since we are in case (II), we
have $|C| \leq n/10$.  Write
\begin{eqnarray}
\sum_{u,v \in X} (f(u)-f(v))^2 &=&
\sum_{u,v \in X} (d(u,S)-d(v,S))^2 \nonumber \\
&\geq & \label{eq:almost}
\sum_{C \in P_0} \sum_{u \in C \cap H_0} \sum_{v \notin C} (d(u,S)-d(v,S))^2.
\end{eqnarray}
So let's estimate $\mathbb E\left[ (d(u,S)-d(v,S))^2 \right]$ for $u \in C \cap H_0$
and $v \notin C$.  Since $u,v$ lie in different clusters, conditioned on
what happens for $v$, we have $d(u,S)$ oscillating randomly between $0$
when $\sigma_C = 0$ and some value greater than $\Delta_2/(4\alpha)$
when $\sigma_C = 1$ (since $u \in H_0$).  It follows that
$\mathbb E\left[ (d(u,S)-d(v,S))^2 \right] \geq \frac{\Delta_2^2}{64 \alpha^2}$.

Plugging this into \eqref{eq:almost} and using $|C| < n/10$ for every $C \in P_0$ yields
\begin{eqnarray*}
\mathbb E \left[\sum_{u,v \in X} (f(u)-f(v))^2\right]
&\geq& \sum_{C \in P_0} \sum_{u \in C \cap H_0} |X \setminus C| \cdot \frac{\Delta_2^2}{64\alpha^2} \\
&\geq & \sum_{C \in P_0} \sum_{u \in C \cap H_0} \frac{9n}{10} \cdot \frac{\Delta_2^2}{64\alpha^2} \\
&=& |H_0| \frac{9n}{10} \frac{\Delta_2^2}{64\alpha^2} \\
&\gtrsim& \left(\frac{n \Delta_2}{\alpha} \right)^2,
\end{eqnarray*}
finishing our proof of \eqref{eq:central}.
\end{proof}

\section{Spectral bounds and balanced separators}
\label{sec:spectral}

We now combine the tools of the previous sections to prove bounds
on the Rayleigh quotients of various graphs.

\subsection{Eigenvalues in bounded degree graphs}
\label{sec:eigenvaluebounds}

Let $G=(V,E)$ be any graph, and set $n = |V|$.
Letting $\lambda_2(G)$ be the second
eigenvalue of the Laplacian of $G$, by the variational
characterization of eigenvalues, we have
\begin{equation}\label{eq:variation}
\lambda_2(G) = \min_{f : V \to \mathbb R} \frac{\sum_{uv \in E} |f(u)-f(v)|^2}{\sum_{u \in V} |f(u)-\bar f|^2}
= 2n \cdot \min_{f :V \to \mathbb R} \frac{\sum_{uv \in E} |f(u)-f(v)|^2}{\sum_{u,v \in V} |f(u)-f(v)|^2}.
\end{equation}
where $\bar f = \frac{1}{n} \sum_{x \in V} f(x)$.

\begin{theorem} \label{thm:eigenbound}
Let $G=(V,E)$ be any graph with maximum degree $d_{\max}$, let $s : V \to \mathbb R_+$
be a non-negative weight function on vertices, and let $d_s$ be the induced semi-metric.
Then,
$$
\lambda_2(G) \lesssim  \frac{d_{\max} n^3 [\alpha(V,d_s)]^2}{\Lambda_s(G)^2}.
$$
\end{theorem}

\begin{proof}
If $d$ is any metric on $V$, then using \eqref{eq:variation} and Theorem \ref{thm:embedding}, we have
$$
\lambda_2(G) \lesssim n \cdot [\alpha(V,d)]^2 \frac{\sum_{uv \in E} d(u,v)^2}{\sum_{u,v \in V} d(u,v)^2}.
$$
Therefore,
\begin{eqnarray}
\lambda_2 &\lesssim& n [\alpha(V,d_s)]^2 \frac{\sum_{uv \in E} d_s(u,v)^2}{\sum_{u,v \in V} d_s(u,v)^2} \nonumber \\
&\leq& n[\alpha(V,d_s)^2] \frac{4 d_{\max} \sum_{v \in V} s(v)^2}{\sum_{u,v\in V} d_s(u,v)^2} \nonumber \\
&\leq & 4 d_{\max} n^3 [\alpha(V,d_s)]^2\frac{\sum_{v \in V} s(v)^2}{\left[\sum_{u,v \in V} d_s(u,v)\right]^2} \nonumber\\
&=& \frac{4 d_{\max} n^3 [\alpha(V,d_s)]^2}{\Lambda_s(G)^2}, \label{eq:ebound}
\end{eqnarray}
where the penultimate inequality follows
from Cauchy-Schwarz.
\end{proof}

\begin{theorem}\label{thm:eigengenus}
If $G=(V,E)$ is a genus $g$ graph with $n=|V|$, then $\lambda_2(G) = O(\frac{d_{\max} g^3}{n})$.
\end{theorem}

\begin{proof}
For any weight function $s : V \to \mathbb R_+$, we have $\alpha(V,d_s) = O(g)$ by Corollary \ref{cor:KPRvertex},
hence \eqref{eq:ebound} yields
$$
\lambda_2(G) \leq \frac{O(d_{\max} g^2 n^3)}{\max_{s : V \to \mathbb R_+} \Lambda_s(G)^2}.
$$
But by Theorems \ref{thm:duality} and \ref{thm:genus}, we have $\max_{s : V \to \mathbb R_+} \Lambda_s(G)^2 = \min_F \vcon_2(F)^2 \geq \Omega(\frac{n^4}{g})$,
where the minimum is over all $K_{n}$-flows in $G$.
\end{proof}

\begin{theorem}\label{thm:eigenminor}
If $G=(V,E)$ is $K_h$-minor-free and $n=|V|$, then $\lambda_2(G) = O(\frac{d_{\max} h^6 \log h}{n})$.
\end{theorem}

\begin{proof}
For any weight function $s : V \to \mathbb R_+$, we have $\alpha(V,d_s) = O(h^2)$ by Corollary \ref{cor:KPRvertex},
hence \eqref{eq:ebound} yields
$$
\lambda_2(G) \leq \frac{O(d_{\max} h^4 n^3)}{\max_{s : V \to \mathbb R_+} \Lambda_s(G)^2}.
$$
But by Theorems \ref{thm:duality} and \ref{thm:exminors}, we have $\max_{s : V \to \mathbb R_+} \Lambda_s(G)^2 = \min_F \vcon_2(F)^2 \geq \Omega(\frac{n^4}{h^2 \log h})$,
where the minimum is over all $K_{n}$-flows $F$ in $G$.
\end{proof}

\begin{remark}
Using Theorem \ref{thm:bartal}, one can give the bounds $\lambda_2(G) = O\left([\max\{ \log n, g \}]^2\right) \frac{d_{\max} g}{n}$ and $\lambda_2(G) = O\left([\max\{\log n, h^2\}]^2\right) \frac{d_{\max} h \log h}{n}$ in Theorems \ref{thm:eigengenus} and \ref{thm:eigenminor}, respectively.  Clearly these bounds are better when $g$ or $h$ grow moderately
fast with $n$.  We suspect that the $O(\cdot)$ part of each bound can be replaced by a universal constant.
The resulting bounds would be tight in the case of genus, and almost tight in the case of $K_h$-minor-free graphs.
It is not clear whether the $\log h$ factor is necessary in general.
\end{remark}

\subsection{Balanced vertex separators}
\label{sec:separators}

For a graph $G=(V,E)$,
a {\em vertex separator} is a partition
$V = A \cup B \cup S$ where
$E(A,B) = \emptyset$.
We define $\alpha(A,B,S) = \frac{|S|}{|A \cup S| \cdot |B \cup S|}$,
and $\alpha(G) = \min \alpha(A,B,S),$
where the minimum is taken over all vertex separators.
We recall the following lemma from \cite[\S 3]{FHL05}.

\begin{lemma}[\cite{FHL05}]\label{lem:FHL}
If $s : V \to \mathbb R_+$ and $d_s$ is the induced metric, then
for any non-expansive map $f : (V,d_s) \to \mathbb R$, there exists
a separator $(A,B,S)$ (found by performing a sweep with $f$),
such that
$$\alpha(A,B,S) \leq \frac{2 \sum_{v \in V} s(v)}{\sum_{u,v \in V} |f(u)-f(v)|}.$$
\end{lemma}

\begin{theorem} \label{thm:separators}
If $G=(V,E)$ is a genus $g$ graph, then $\alpha(G) = O(n^{-3/2} g^{3/2})$.
If $G=(V,E)$ is $K_h$-minor free, then $\alpha(G) = O( n^{-3/2} h^3 \sqrt{\log h})$.
Therefore such graphs have $O(g^{3/2} \sqrt{n})$ and $O(h^{3} \sqrt{\log h} \sqrt{n})$-sized
$(\frac13, \frac23)$-balanced separators, respectively.
\end{theorem}

\begin{proof}
We prove the theorem only for a $K_h$-minor free graph $G$.
$$
\min_{s : V \to \mathbb R_+} \frac{\sum_{v \in V} s(v)}{\sum_{u,v \in V} d_s(u,v)}
\leq \sqrt{n} \cdot \min_{s : V \to \mathbb R_+} \frac{\sqrt{\sum_{v \in V} s(v)^2}}{\sum_{u,v \in V} d_s(u,v)}
= \frac{\sqrt{n}}{\max_{s : V \to \mathbb R_+} \Lambda_s(G)} = O\left(\frac{h \sqrt{\log h}}{n^{3/2}}\right),
$$
by Theorems \ref{thm:duality} and \ref{thm:exminors}.

Let $s : V \to \mathbb R_+$ be the minimizer of the left-hand side, and
apply Theorem \ref{thm:embedding} to $(V,d_s)$ with $p=1$.
Since $\alpha(V,d_s) = O(h^2)$ by Corollary \ref{cor:KPRvertex}, this
yields a non-expansive map $f : (V,d_s) \to \mathbb R$ with
$$
\frac{\sum_{v \in V} s(v)}{\sum_{u,v \in V} |f(u)-f(v)|} \leq O(h^2) \frac{\sum_{v \in V} s(v)}{\sum_{u,v \in V} d_s(u,v)} = O\left(\frac{h^3 \sqrt{\log h}}{n^{3/2}}\right).
$$
We apply Lemma \ref{lem:FHL} to achieve $\alpha(G) = O(n^{-3/2} h^3 \sqrt{\log h})$.
Finally, we note that the passage from this bound on $\alpha(G)$ to balanced separators of size
$O(h^3 \sqrt{\log h} \sqrt{n})$ uses a standard recursive quotient cut algorithm;
see e.g. \cite[\S 6]{FHL05}.
\end{proof}

\subsection{Geometric graphs}
\label{sec:geographs}

In practice, spectral methods are applied to graphs arising from a
variety of geometric settings, not limited to surfaces of fixed genus.
Miller \emph{et al.}\ considered $k$-ply neighborhood systems,
$k$-nearest neighbor graphs, and well-shaped finite element meshes in
any fixed dimension \cite{mt-sphere-packings,mt-mesh}.
They used geometric techniques to efficiently find small ratio cuts
for these classes of graphs.
Spielman and Teng give bounds for the
second eigenvalue of these graphs,
thus showing that these cuts
can be recovered by spectral partitioning \cite{spielman-teng}.

While none of these graph families exclude a fixed set of minors, some
of them have been shown to lack small minors at small depth. We can
adapt the proofs from the preceding subsections to this setting using
the following lemma.

\begin{lemma}
  \label{lem:eigenshallow} Let $G = (V,E)$ have constant maximum
  degree and exclude a $K_h$-minor at depth $L$, where $h = L^p$ for
  some constant $p$ and $|V| = n$. We have,
  \begin{itemize}
  \item $\lambda_2(G) = \tilde O(n^{-1/(1+p)})$.
  \item $G$ has a vertex separator $(A,B,S)$ satisfying $\alpha(A,B,S)
    = \tilde O(n^{-1-\frac{1}{2+2p}})$. It also has
    $(1/3,2/3)$-balanced separators of size $\tilde O(n^{1-\frac{1}{2+2p}})$.
  \end{itemize}
\end{lemma}
\begin{proof}
  For any weight function $s\colon V \to \R_+$, we have $\alpha(V,
  d_s) = \tilde O(1)$ by Theorem~\ref{thm:bartal}.  By
  Theorems~\ref{thm:duality} and \ref{thm:shallowcon}, we have
  \[
  \max_{s\colon V \to \R_+} \Lambda_s(G) =
  \min_F \vcon_2(F)
  \gtrsim \min\left( \frac{n^2}{h \sqrt{\log h}},
    n^{3/2} L\right)
  = \tilde\Omega\left(n^{\frac{3}{2} + \frac{1}{2+2p}}\right)
  \]
  where the last equality follows by setting $L=n^{1/(2+2p)}$, and the
  $\tilde O(\cdot)$ and $\tilde\Omega(\cdot)$ notations hide $\mathrm{poly}(\log n)$ factors. The
  eigenvalue bound follows immediately from an application of
  Theorem~\ref{thm:eigenbound}.

  As in the proof of Theorem~\ref{thm:separators}, we have
  \[
  \min_{s\colon V\to \R_+} \frac{\sum_{v\in V} s(v)}{\sum_{u,v\in V}d_s(u,v)}
  \leq \frac{\sqrt n}{\max_{s\colon V\to \R_+} \Lambda_s(G)}
  = \tilde O\left(n^{-1-\frac{1}{2+2p}}\right).
  \]
  This gives us a non-expansive map $f \colon (V,d_s) \to \R$ with
  \[
  \frac{\sum_{v\in V} s(v)}{\sum_{u,v\in V}|f(u)-f(v)|}
  \leq \tilde O(1) \frac{\sum_{v\in V} s(v)}{\sum_{u,v\in V}d_s(u,v)}
  = \tilde O\left(n^{-1-\frac{1}{2+2p}}\right).
  \]
  An application of Lemma~\ref{lem:FHL} establishes the bound on
  $\alpha(A,B,S)$, and recursive quotient cuts can be used to find the
  corresponding balanced separators.
\end{proof}

\begin{theorem}
  \label{thm:eigengeom} Let $G = (V,E)$ be a graph with $n$
  vertices and constant maximum degree.
  \begin{itemize}
  \item If $G$ is a simplicial graph in $d$ dimensions with constant
    aspect ratio (see \cite{miller-thurston} for a detailed
    definition), then $\lambda_2(G) = \tilde O(n^{-1/d})$ and $G$ has
    balanced separators of size $\tilde O(n^{1-\frac1{2d}})$.
  \item If $G$ is an \emph{arbitrary} $k$-nearest neighbor graph in
    $d$ dimensions, then $\lambda_2(G) = \tilde O(n^{-1/(1+d)})$ and
    $G$ has balanced separators of size $\tilde
    O(n^{1-\frac{1}{2+2d}})$.
  \item If $G$ is a $d$-dimensional grid, then $\lambda_2(G) = \tilde
    O(n^{-2/(2+d)})$ and $G$ has balanced separators of size $\tilde
    O(n^{1-\frac{1}{2+d}})$. This is also true with high probability
    when $G$ is the relative neighborhood graph, the Delaunay diagram,
    or the $k$-nearest neighbor graph of a \emph{random} point set in
    $d$ dimensions.
  \end{itemize}
\end{theorem}

\begin{proof}
  The results follow from the corresponding bounds on excluded shallow
  minors. From Plotkin \emph{et al.} \cite{shallow-minors}, we have $h
  = \Omega_d(L^{d-1})$ for simplicial graphs and $h =
  \Omega_d(L^{d/2})$ for grids. From Teng \cite{teng}, we have $h =
  \Omega(L^d)$ for arbitrary $k$-nearest neighbor graphs and, with
  high probability, $h = \Omega(L^{d/2})$ for the relative
  neighborhood graph, the Delaunay diagram, and the $k$-nearest
  neighbor graph of a random point set.
\end{proof}

\begin{remark}
  Spielman and Teng \cite{spielman-teng} prove that
  $k$-nearest-neighbor graphs and well-shaped meshes have $\lambda_2$ of
  value $O(n^{-2/d})$ and balanced separators of ratio
  $O(n^{-1/d})$. Their results are better than ours by a square;
  we suspect this is due to the non-tightness of the
  bounds on shallow excluded minors for these graph families.
\end{remark}

\remove{

\section{Discussion and open problems}

\begin{enumerate}
\item Clearly it would be nice if our approach could give $\lambda_2(G) \leq O(\frac{g d_{\max}}{n})$,
which would be optimal.  This bound matches 
\item Higher eigenvalues
\item Geometric graphs
\item Improve KPR yields almost-optimal separator theorem
\end{enumerate}
}

\subsection*{Acknowledgements}
We thank Oded Schramm for helpful pointers
to the literature on discrete conformal mappings.

\bibliographystyle{abbrv}
\bibliography{spectral}

\def\cprime{$'$} \def\cprime{$'$}
\begin{thebibliography}{10}

\bibitem{acns}
M.~Ajtai, V.~Chv\'atal, M.~Newborn, and E.~Szemer\'edi.
\newblock Crossing-free subgraphs.
\newblock In A.~Kotzig, A.~Rosa, G.~Sabidussi, and J.~Turgeon, editors, {\em
  Theory and Practice of Combinatorics: A Collection of Articles Honoring Anton
  Kotzig on the Occasion of His Sixtieth Birthday.}, volume~12 of {\em Annals
  of discrete mathematics}. North-Holland, Amsterdam, 1982.

\bibitem{AM85}
N.~Alon and V.~D. Milman.
\newblock {$\lambda\sb 1,$} isoperimetric inequalities for graphs, and
  superconcentrators.
\newblock {\em J. Combin. Theory Ser. B}, 38(1):73--88, 1985.

\bibitem{AST90}
N.~Alon, P.~Seymour, and R.~Thomas.
\newblock A separator theorem for nonplanar graphs.
\newblock {\em J. Amer. Math. Soc.}, 3(4):801--808, 1990.

\bibitem{AK95}
C.~J. Alpert and A.~B. Kahng.
\newblock Recent directions in netlist partitioning: A survey.
\newblock {\em Integration: The VLSI J.}, 19:1--81, 1995.

\bibitem{Bartal96}
Y.~Bartal.
\newblock Probabilistic approximations of metric space and its algorithmic
  application.
\newblock In {\em 37th Annual Symposium on Foundations of Computer Science},
  pages 183--193, Oct. 1996.

\bibitem{BCE80}
B.~Bollob{\'a}s, P.~A. Catlin, and P.~Erd{\H{o}}s.
\newblock Hadwiger's conjecture is true for almost every graph.
\newblock {\em European J. Combin.}, 1(3):195--199, 1980.

\bibitem{Bourgain85}
J.~Bourgain.
\newblock On {L}ipschitz embedding of finite metric spaces in {H}ilbert space.
\newblock {\em Israel J. Math.}, 52(1-2):46--52, 1985.

\bibitem{boyd}
S.~Boyd and L.~Vandenberghe.
\newblock {\em Convex optimization}.
\newblock Cambridge University Press, Cambridge, 2004.

\bibitem{CKR01}
G.~Calinescu, H.~Karloff, and Y.~Rabani.
\newblock Approximation algorithms for the 0-extension problem.
\newblock In {\em Proceedings of the 12th Annual ACM-SIAM Symposium on Discrete
  Algorithms}, pages 8--16, Philadelphia, PA, 2001. SIAM.

\bibitem{CSZ93}
P.~K. Chan, M.~Schlag, and J.~Zien.
\newblock Speactral $k$-way ratio cut partitioning and clustering.
\newblock In {\em Proceedings of the Symposium on Integrated Systems}, 1993.

\bibitem{CR87}
T.~F. Chan and D.~C. Resasco.
\newblock A framework for the analysis and construction of domain decomposition
  preconditioners.
\newblock Technical Report CAM-87-09, UCLA, 1987.

\bibitem{CS93}
T.~F. Chan and B.~Smith.
\newblock Domain decomposition and multigrid algorithms for elliptic problems
  on unstructured meshes.
\newblock {\em Contemp. Math.}, pages 1--14, 1993.

\bibitem{Cheeger70}
J.~Cheeger.
\newblock A lower bound for the smallest eigenvalue of the {L}aplacian.
\newblock In {\em Problems in analysis (Papers dedicated to Salomon Bochner,
  1969)}, pages 195--199. Princeton Univ. Press, Princeton, N. J., 1970.

\bibitem{Chung97}
F.~R.~K. Chung.
\newblock {\em Spectral graph theory}, volume~92 of {\em CBMS Regional
  Conference Series in Mathematics}.
\newblock Published for the Conference Board of the Mathematical Sciences,
  Washington, DC, 1997.

\bibitem{DiestelBook}
R.~Diestel.
\newblock {\em Graph theory}, volume 173 of {\em Graduate Texts in
  Mathematics}.
\newblock Springer-Verlag, Berlin, third edition, 2005.

\bibitem{Duffin62}
R.~J. Duffin.
\newblock The extremal length of a network.
\newblock {\em J. Math. Anal. Appl.}, 5:200--215, 1962.

\bibitem{FHL05}
U.~Feige, M.~T. Hajiaghayi, and J.~R. Lee.
\newblock Improved approximation algorithms for minimum-weight vertex
  separators.
\newblock In {\em 37th Annual ACM Symposium on Theory of Computing}. ACM, 2005.
\newblock To appear, {\em SIAM J. Comput.}

\bibitem{Vega83}
W.~Fernandez de~la Vega.
\newblock On the maximum density of graphs which have no subcontraction to
  {$K\sp{s}$}.
\newblock {\em Discrete Math.}, 46(1):109--110, 1983.

\bibitem{GHT84}
J.~R. Gilbert, J.~P. Hutchinson, and R.~E. Tarjan.
\newblock A separator theorem for graphs of bounded genus.
\newblock {\em J. Algorithms}, 5(3):391--407, 1984.

\bibitem{GY99}
A.~Grigor{\cprime}yan and S.-T. Yau.
\newblock Decomposition of a metric space by capacitors.
\newblock In {\em Differential equations: La Pietra 1996 (Florence)}, volume~65
  of {\em Proc. Sympos. Pure Math.}, pages 39--75. Amer. Math. Soc.,
  Providence, RI, 1999.

\bibitem{GM98}
S.~Guattery and G.~L. Miller.
\newblock On the quality of spectral separators.
\newblock {\em SIAM J. Matrix Anal. Appl.}, 19(3):701--719 (electronic), 1998.

\bibitem{HK92}
L.~Hagen and A.~B. Kahng.
\newblock New spectral methods for ratio cut partitioning and clustering.
\newblock {\em IEEE Trans. Computer-Aided Des.}, 11(9):1074--1085, 1992.

\bibitem{HS95}
Z.-X. He and O.~Schramm.
\newblock Hyperbolic and parabolic packings.
\newblock {\em Discrete Comput. Geom.}, 14(2):123--149, 1995.

\bibitem{Hersch70}
J.~Hersch.
\newblock Quatre propri\'et\'es isop\'erim\'etriques de membranes sph\'eriques
  homog\`enes.
\newblock {\em C. R. Acad. Sci. Paris S\'er. A-B}, 270:A1645--A1648, 1970.

\bibitem{Kelner06}
J.~A. Kelner.
\newblock Spectral partitioning, eigenvalue bounds, and circle packings for
  graphs of bounded genus.
\newblock {\em SIAM J. Comput.}, 35(4):882--902 (electronic), 2006.

\bibitem{KPR93}
P.~N. Klein, S.~A. Plotkin, and S.~Rao.
\newblock Excluded minors, network decomposition, and multicommodity flow.
\newblock In {\em Proceedings of the 25th Annual ACM Symposium on Theory of
  Computing}, pages 682--690, 1993.

\bibitem{K93}
N.~Korevaar.
\newblock Upper bounds for eigenvalues of conformal metrics.
\newblock {\em J. Differential Geom.}, 37(1):73--93, 1993.

\bibitem{Kost82}
A.~V. Kostochka.
\newblock The minimum {H}adwiger number for graphs with a given mean degree of
  vertices.
\newblock {\em Metody Diskret. Analiz.}, (38):37--58, 1982.

\bibitem{KLMN05}
R.~Krauthgamer, J.~R. Lee, M.~Mendel, and A.~Naor.
\newblock Measured descent: a new embedding method for finite metrics.
\newblock {\em Geom. Funct. Anal.}, 15(4):839--858, 2005.

\bibitem{Leighton}
F.~T. Leighton.
\newblock {\em Complexity Issues in VLSI}.
\newblock MIT Press, 1983.

\bibitem{LT79}
R.~J. Lipton and R.~E. Tarjan.
\newblock A separator theorem for planar graphs.
\newblock {\em SIAM J. Appl. Math.}, 36(2):177--189, 1979.

\bibitem{Lo06}
L.~Lov{\'a}sz.
\newblock Graph minor theory.
\newblock {\em Bull. Amer. Math. Soc. (N.S.)}, 43(1):75--86 (electronic), 2006.

\bibitem{LovGRep07}
L.~Lov\'asz.
\newblock Geometric representations of graphs.
\newblock {\small \verb|http://www.cs.elte.hu/~lovasz/geomrep.pdf|}, 2007.

\bibitem{Mih89}
G.~Mihail.
\newblock Conductance and convergence of {M}arkov chains---a combinatorial
  treatment of expanders.
\newblock In {\em {IEEE} Symposium on Foundations of Computer Science}, pages
  14--23, 1989.

\bibitem{mt-sphere-packings}
G.~L. Miller, S.-H. Teng, W.~Thurston, and S.~A. Vavasis.
\newblock Separators for sphere-packings and nearest neighbor graphs.
\newblock {\em J. ACM}, 44(1):1--29, 1997.

\bibitem{mt-mesh}
G.~L. Miller, S.-H. Teng, W.~Thurston, and S.~A. Vavasis.
\newblock Geometric separators for finite-element meshes.
\newblock {\em SIAM Journal on Scientific Computing}, 19(2):364--386, 1998.

\bibitem{miller-thurston}
G.~L. Miller and W.~Thurston.
\newblock Separators in two and three dimensions.
\newblock In {\em STOC '90: Proceedings of the twenty-second annual ACM
  symposium on Theory of computing}, pages 300--309, New York, NY, USA, 1990.
  ACM.

\bibitem{PRS94}
S.~Plotkin, S.~Rao, and W.~D. Smith.
\newblock Shallow excluded minors and improved graph decompositions.
\newblock In {\em Proceedings of the Fifth Annual ACM-SIAM Symposium on
  Discrete Algorithms (Arlington, VA, 1994)}, pages 462--470, New York, 1994.
  ACM.

\bibitem{shallow-minors}
S.~Plotkin, S.~Rao, and W.~D. Smith.
\newblock Shallow excluded minors and improved graph decompositions.
\newblock In {\em SODA '94: Proceedings of the fifth annual ACM-SIAM symposium
  on Discrete algorithms}, pages 462--470, Philadelphia, PA, USA, 1994. Society
  for Industrial and Applied Mathematics.

\bibitem{PSW92}
A.~Pothen, H.~D. Simon, and L.~Wang.
\newblock Spectral nested disection.
\newblock Technical report CS-92-01, Pennsylvania State University, Department
  of Computer Science, 1992.

\bibitem{Rab03}
Y.~Rabinovich.
\newblock On average distortion of embedding metrics into the line and into
  {$L_1$}.
\newblock In {\em 35th Annual ACM Symposium on Theory of Computing}. ACM, 2003.

\bibitem{Rao99}
S.~Rao.
\newblock Small distortion and volume preserving embeddings for planar and
  {E}uclidean metrics.
\newblock In {\em Proceedings of the 15th Annual Symposium on Computational
  Geometry}, pages 300--306, New York, 1999. ACM.

\bibitem{SY94}
R.~Schoen and S.-T. Yau.
\newblock {\em Lectures on differential geometry}.
\newblock Conference Proceedings and Lecture Notes in Geometry and Topology, I.
  International Press, Cambridge, MA, 1994.

\bibitem{Schramm93}
O.~Schramm.
\newblock Square tilings with prescribed combinatorics.
\newblock {\em Israel J. Math.}, 84(1-2):97--118, 1993.

\bibitem{SM00}
J.~Shi and J.~Malik.
\newblock Normalized cuts and image segmentation.
\newblock {\em IEEE Transactions on Pattern Analysis and Machine Intelligence},
  22(8):888--905, 2000.

\bibitem{Simon91}
H.~D. Simoh.
\newblock Partitioning of unstructured problems for parallel processing.
\newblock {\em Comput. Syst. Eng.}, 2(2):135--148, 1991.

\bibitem{spielman-teng}
D.~A. Spielman and S.-H. Teng.
\newblock Spectral partitioning works: Planar graphs and finite element meshes.
\newblock {\em Linear Algebra and its Applications: Special Issue in honor of
  Miroslav Fiedler}, 421(2--3):284--305, March 2007.

\bibitem{teng}
S.-H. Teng.
\newblock Combinatorial aspects of geometric graphs.
\newblock {\em Computational Geometry}, pages 277--287, 1998.

\bibitem{Thom84}
A.~Thomason.
\newblock An extremal function for contractions of graphs.
\newblock {\em Math. Proc. Cambridge Philos. Soc.}, 95(2):261--265, 1984.

\bibitem{Williams90}
R.~D. Williams.
\newblock Performance of dynamic load balancing algorithms for unstructured
  mesh calculations.
\newblock Technical Report, California Institute of Technology, 1990.

\bibitem{YY80}
P.~C. Yang and S.~T. Yau.
\newblock Eigenvalues of the {L}aplacian of compact {R}iemann surfaces and
  minimal submanifolds.
\newblock {\em Ann. Scuola Norm. Sup. Pisa Cl. Sci. (4)}, 7(1):55--63, 1980.

\end{thebibliography}

\end{document}